\documentclass[a4paper]{llncs} 

\usepackage[utf8]{inputenc}

\usepackage{amsmath}
\usepackage{amsfonts}
\usepackage{color}

\usepackage{tikz}

\usepackage[all,cmtip]{xy}

\title{Shared Information -- New Insights and Problems in Decomposing Information in Complex Systems}

\author{Nils Bertschinger \and Johannes Rauh \and Eckehard Olbrich \and Jürgen Jost}
\institute{MPI für Mathematik in den Naturwissenschaften \\
Inselstra\ss{}e 22, 04109 Leipzig \\
\email{\{bertschi,jrauh,olbrich,jost\}@mis.mpg.de}
}

\newcommand{\set}[1]{\mathcal{#1}}

\newcommand{\Pcal}{\ensuremath{\mathcal{P}}}
\newcommand{\Rb}{\ensuremath{\mathbb{R}}}

\DeclareMathOperator{\argmin}{argmin}

\newcommand{\sh}{\ensuremath{|}}

\begin{document}

\maketitle

\begin{abstract}
  How can the information that a set $\{X_{1},\dots,X_{n}\}$ of random variables contains about another random variable $S$
  be decomposed?  To what extent do different subgroups
  provide the same, i.e.~shared or redundant, information, carry unique
  information or interact for the emergence of synergistic information?
  
  Recently Williams and Beer proposed such a decomposition based on natural
  properties for shared information.
  While these properties 
  fix the structure of the decomposition, they do not uniquely specify the
  values of the different terms.  Therefore, we investigate additional
  properties such as strong symmetry and left monotonicity. We find that strong
  symmetry is incompatible with the properties proposed by Williams and
  Beer. Although left monotonicity is a very natural property for an information
  measure it is not fulfilled by any of the proposed measures.

  We also study a geometric framework for information decompositions and ask
  whether it is possible to represent shared information by a family of
  posterior distributions.

  Finally, we draw connections to the notions of shared knowledge and common
  knowledge in game theory.  While many people believe that independent
  variables cannot share information, we show that in game theory independent
  agents can have shared knowledge, but not common knowledge.
  We conclude that intuition and heuristic arguments do not suffice when arguing
  about information.
\end{abstract}

\section{Introduction}

The field of complex systems investigates systems which are composed of many components or sub-systems. Such a system is considered as complex if these components interact in intricate ways and exhibit dependencies at all scales. Informally,
complex systems are often described in terms of information that is exchanged between components. Thus, information
theory is a natural tool to study complex systems. 

As an example from neural coding, consider two neurons which provide
information about some stimulus.  Many scientists have tried to
uncover whether both neurons provide redundant information about the
stimulus or act synergetically, i.e.~provide information which can
only be recovered when the joint response of both cells is recorded
simultaneously~\cite{Schneidman2003,Latham2005}.  Similarly, one could
ask for the unique information of each response, i.e.~information that
can be obtained from one of the cells, but not the other. For example,
the brain separates visual information into the where and what
pathways~\cite{Pessoa2008} which potentially provide unique information with respect to
each other.  Another way to explain the intuition on
how information can be decomposed, is to consider two agents which are
interrogated about certain topics. For example, assume that one agent is an expert in
physics and biology, whereas the other one has studied art and
biology. In this case, both agents could answer questions about
biology being their shared topic. Furthermore, each agent has
additional unique information about physics and art,
respectively. Considering their joint responses an interrogator might
be able to draw interesting connections between art and physics
none of the agents is aware of.  This would correspond to the
synergetic information in this case.

In general, when considering more than two random variables, there may be
different combinations of shared, unique and synergistic information, depending
on how the information is distributed among the random variables.  The total
mutual information $I(S : X_{1},\dots,X_{n})$ should then be a sum of different
terms with a well-defined interpretation.  At the moment, it is not clear how
many such terms are necessary in the general case of $n$ interacting elements.
Williams and Beer recently proposed one such decomposition, which they call
\emph{partial information (PI) decomposition}~\cite{WB}.  This decomposition is
naturally derived from simple intuitive properties that such a decomposition
should satisfy.

Before explaining the construction of Williams and
Beer, we first have a look at the case of $n=2$ explanatory variables
in Section~\ref{sec:two-variables}.
In Section~\ref{sec:natur-prop-SI} we discuss natural properties that such a
decomposition should satisfy and, following Williams and Beer, use these
properties to derive the PI decomposition.  In
Section~\ref{sec:further-properties} we propose additional properties that
relate the values of shared information in situations where we ask for
information about different variables.  In Section~\ref{sec:Imin} we discuss the
measure $I_{\min}$ proposed by Williams and Beer and compare it to another
function~$I_{I}$, i.e. the minimum of the pairwise mutual informations.  We show that the function $I_{\min}$ may decrease when we ask
for information about a larger variable.
In Section~\ref{sec:three-variables}, we study the case for three variables.  We
show that it is difficult to assign intuitively plausible values to all partial
information terms, even in the simple XOR-example.  Using this example we show
that the structure of the PI lattice is incompatible with a symmetry property
which we call strong symmetry.

In Section~\ref{sec:geometry} we propose a geometric picture for information
decomposition. This view provides an appealing mathematical structure and
provides additional insights into the structure of information.  Within this geometric framework, we compare
our ideas to the measures proposed in \cite{WB} and \cite{HarderSalgePolani12:Bivariate_Redundant_Information}.
Then, in
Section~\ref{sec:game-theory}, we study the game theoretic notions of shared and
common knowledge that are used to describe epistemic states of multi-agent
systems, and we discuss how these notions are related to the problem of
decomposing information.  We conclude with an outlook on the possibility
of a general decomposition of information.

\section{The Case of Two Variables}
\label{sec:two-variables}

First, we fix the notation and recall some basic definitions from information theory~\cite{CoverThomas}. We assume that
a system consists of $N$ components $X_1, \ldots, X_N$.  For simplicity we assume that the set of possible states
$\set{X}_i$ that a component $X_{i}$ can be in is finite.  Thus, the set of all possible states for the whole system is
given by $\set{X}_1^N = \times_{i = 1}^N \set{X}_i$.

Given a probability distribution $p$ on $\set{X}_1^N$, the $X_i$ become random variables.
%
Mutual information between two random variables $X$ and $Y$ quantifies the information about $Y$ that is gained by knowing $X$ and vice versa. It can be defined as
\begin{equation} 
  I(X: Y) = \sum_{y \in \set{Y}} p(y) D(p(X|y) \| p(X))
  \label{equ:KL}
\end{equation}
where $D(p(X|y) \| p(X)) = \sum_{x \in \set{X}} p(x|y) \log_2 \frac{p(x|y)}{p(x)}$ is the Kullback-Leibler (KL)
divergence between $p(X|y)$ and $p(X)$\footnote{Here, $p(X)$ denotes the probability distribution of the random variable
  $X$. When referring to the probability of a particular outcome $x \in \set{X}$ of this random variable, we write
  $p(x)$.}. The KL divergence is often considered as a distance between probability distributions even though it is not
a metric. But, like a metric, it vanishes if and only if the two distributions are identical. It can also be interpreted
as an information gain: if one finds out that $Y=y$ then $D(p(x|y)\|p(x))$ bits of information are gained about $X$.
It is well known that the mutual information is symmetric and vanishes if and only if $X$ and $Y$ are independent.

Consider now three random variables $X_1, X_2$ and~$S$.  The (total) mutual information $I(S; (X_1, X_2))$ quantifies
the total information that is gained about $S$ if the outcome of $X_1$ and $X_2$ is known.  How do $X_1$ and $X_2$ contribute to this information?

For two explanatory variables, we expect four contributions to $I(S : X_{1}X_{2})$:
\begin{equation} 
 I(S: X_1 X_2) = SI(S: X_1; X_2) + UI(S: X_1 \setminus X_2) + UI(S: X_2 \setminus X_1) + CI(S: X_1; X_2)
  \label{equ:decomp}
\end{equation}
The shared (redundant) information $SI(S: X_1; X_2)$, the unique
informations $UI$ and the complementary (synergistic) information
$CI(S: X_1;X_2)$.
Intuition tells us that the individual
mutual informations that are provided by $X_1$ and $X_2$ should decompose as
\begin{equation}
   \label{equ:decomp2}
  \begin{aligned}
 I(S: X_1)  &= SI(S: X_1; X_2) + UI(S: X_1 \setminus X_2) \\
    I(S: X_2) & = SI(S: X_1; X_2) + UI(S: X_2 \setminus X_1) \;.
  \end{aligned}
\end{equation}
Using the full decomposition \eqref{equ:decomp} and the chain rule of
mutual information~\cite{CoverThomas} we find that the conditional
informations correspond to unique and complementary information,
e.g.~$I(S:X_1|X_2) = UI(S:X_1 \setminus X_2) + CI(S:X_1;X_2)$.  Furthermore, we
recover the fact that the co-information $I_{Co}$~\cite{Bell2003} contemplates
shared and complementary information, i.e.
\begin{equation}
  \label{equ:coinformation-diff}
  I_{Co}(S:X_1: X_2) := I(S: X_1|X_2) - I(S:X_1) = CI(S: X_1;X_2) - SI(S: X_1;X_2)
\end{equation}

Unfortunately, the three linear equations~\eqref{equ:decomp}
and~\eqref{equ:decomp2} do not completely specify the four functions
on the right hand side of~\eqref{equ:decomp}.
To determine the decomposition~\eqref{equ:decomp} it is sufficient to
define one of the functions $SI$, $UI$ and $CI$.  It seems to be a
difficult task to come up with a reasonable and well-motivated
definition of $SI$ such that the induced definitions of $UI$ and $CI$
via equations~\eqref{equ:decomp} and~\eqref{equ:decomp2} are
non-negative.  The same is true when trying to find formulas for $UI$
or $CI$.  Note that any definition of the unique information fixes two
of the terms in~\eqref{equ:decomp}.  This leads to the consistency
condition
\begin{equation}
  \label{equ:UI-consistency}
  I(S:X_1) + UI(S:X_2\setminus X_1) = I(S:X_2) + UI(S:X_1\setminus X_2),
\end{equation}
which resembles the chain rule. Indeed, $UI(S:X_1\setminus X_2)$ can be
considered as a version of conditional information which does not contain the
complementary information\footnote{A related notion has been developed in the
  context of cryptography to quantify the secret information. Although the
  secret information has a clear operational interpretation it cannot be
  computed directly, but is upper bounded by the \emph{intrinsic mutual
    information} $I(S: X_1 \downarrow X_2)$
  \cite{MaurerWolf97,ChristandlRennerWolf03}. Unfortunately, the intrinsic
  mutual information does not obey the consistency
  condition~\eqref{equ:UI-consistency}, and hence it cannot be interpreted as
  unique information in our sense.}.

Apart from the problem of finding formulas for $SI$, $UI$ and $CI$, a
second problem is how to generalize the decomposition~\eqref{equ:decomp} to more
than two explanatory variables.  A possible solution to
both problems was recently proposed by Williams and Beer.

\section{Natural Properties of Shared Information and the Partial Information Lattice}
\label{sec:natur-prop-SI}

Williams and Beer~\cite{WB} base their construction of a non-negative decomposition
of $I(S: X_{1}\dots X_{n})$ on the notion of redundancy or shared
information.  Let
$\mathbf{A}_{1},\dots,\mathbf{A}_{k}\subseteq\{X_{1},\dots,X_{n}\}$,
and denote by $I_{\cap}(S:\mathbf{A}_{1};\dots;\mathbf{A}_{k})$ the
information about $S$ that is shared among the random variables in the
sets $\mathbf{A}_{1},\dots,\mathbf{A}_{k}$.  
It is natural to demand that $I_{\cap}$ satisfy the following properties:
\begin{itemize}
\item[\textbf{(GP)}] $I_{\cap}(S : \mathbf{A}_{1};\dots;\mathbf{A}_{k})\ge 0$.
  \hfill\emph{(global positivity)}
\item[\textbf{(S$_{0}$)}] $I_{\cap}(S: \mathbf{A}_{1};\dots;\mathbf{A}_{k})$ is symmetric in $\mathbf{A}_{1},\dots,\mathbf{A}_{k}$.
  \hfill\emph{(weak symmetry)}
\item[\textbf{(I)}] $I_{\cap}(S: \mathbf{A}) = I(S:\mathbf{A})$ equals the mutual information of $S$ and $\mathbf{A}$.

  \hfill\mbox{\emph{(self-redundancy)}}
\item[\textbf{(M)}] $I_{\cap}(S: \mathbf{A}_{1};\dots;\mathbf{A}_{k}) \le I_{\cap}(S: \mathbf{A}_{1};\dots;\mathbf{A}_{k-1})$, with equality if $\mathbf{A}_{k-1}$ is
  a subset of $\mathbf{A}_{k}$.
  \hfill\emph{(monotonicity)}
\end{itemize}
The properties \textbf{(S$_{0}$)}, \textbf{(I)} and \textbf{(M)} have been proposed as axioms of shared information by
Williams and Beer in~\cite{WB}.  As Williams and Beer observe, \textbf{(GP)} is a consequence of the other properties.
Here we like to state it as a separate property, since we want to discuss what happens if we drop or relax some of these
properties.

The properties \textbf{(S$_{0}$)} and \textbf{(M)} imply that it is sufficient
to define the function $I_{\cap}(S:\mathbf{A}_{1};\dots;\mathbf{A}_{k})$ in
the case that $\mathbf{A}_{i}\not\subseteq \mathbf{A}_{j}$ for all $i\neq j$.  A family of sets
$\mathbf{A}_{1},\dots,\mathbf{A}_{k}$ with this property is called an \emph{anti-chain}.  The anti-chains form a
lattice with respect to the partial order defined by $(\mathbf{B}_{1},\dots,\mathbf{B}_{k}) \le (\mathbf{A}_{1},\dots,\mathbf{A}_{l})$
if and only if for each $i=1,\dots,l$ there exists $j\in\{1,\dots,k\}$ such that
$\mathbf{B}_{j}\subseteq\mathbf{A}_{i}$.
If $S$ is fixed, then \textbf{(S$_{0}$)} and \textbf{(M)} imply that $I_{\cap}(S:\cdot)$ is a monotone function on the
lattice of anti-chains of $\{X_{1},\dots,X_{n}\}$: If $(\mathbf{B}_{1},\dots,\mathbf{B}_{k}) \le
(\mathbf{A}_{1},\dots,\mathbf{A}_{l})$, then
\begin{equation*}
  I_{\cap}(S: \mathbf{B}_{1},\dots,\mathbf{B}_{k}) = I_{\cap}(S: \mathbf{B}_{1},\dots,\mathbf{B}_{k},\mathbf{A}_{1},\dots,\mathbf{A}_{k})  \le  I_{\cap}(S: \mathbf{A}_{1},\dots,\mathbf{A}_{l}).
\end{equation*}
This lattice is also called the \emph{partial information (PI) lattice}.  In this paper, we focus on the case of
two or three random variables, and the corresponding lattices are depicted in Figures~\ref{fig:2-vars}
and~\ref{fig:3-vars}.
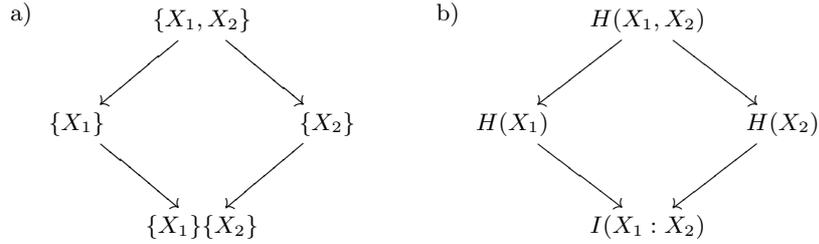
\begin{figure}
  \centering
  \begin{minipage}{0.45\linewidth}
    a)
    \xymatrix@C=1em{
      & \{X_1,X_2\} \ar[ld] \ar[rd] & \\
      \{X_1\} \ar[rd] & & \{X_2\} \ar[ld] \\
      & \{X_1\}\{X_2\} & }
  \end{minipage}
  \begin{minipage}{0.5\linewidth}
    b)
    \xymatrix@C=1em{
      & H(X_{1},X_{2}) \ar[ld] \ar[rd] & \\
      H(X_{1}) \ar[rd] & & H(X_{2}) \ar[ld] \\
      & I(X_{1}:X_{2}) &
    }
  \end{minipage}
  
  \caption{The PI lattice for two random variables. a) The sets corresponding to the nodes in the lattice. b) The
    redundancies at the nodes for $S=\{X_{1},X_{2}\}$, assuming strong symmetry (see \textbf{(S$_1$)} in Section~\ref{sec:further-properties}).}
\label{fig:2-vars}
\end{figure}

Properties \textbf{(M)} and \textbf{(I)} imply $I_{\cap}(S:\mathbf{A}_{1};\dots;\mathbf{A}_{k}) \le I_{\cap}(S:
\mathbf{A}_{1}) = I(S: \mathbf{A}_{1}) \le I(S: X_{1}\dots X_{n})$.  To obtain a decomposition of this total mutual
information, we need to associate to each element of the PI lattice a ``local quantity'' $I_{\partial}$ in such a way
that
\begin{equation*}
  I_{\cap}(S:\mathbf{A}_{1};\dots;\mathbf{A}_{k}) = \sum_{(\mathbf{B}_{1},\dots,\mathbf{B}_{l})\le(\mathbf{A}_{1},\dots,\mathbf{A}_{k})} I_{\partial}(S:\mathbf{B}_{1},\dots,\mathbf{B}_{l}).
\end{equation*}
One can show, using the notion of a Möbius inversion, that such a function $I_{\partial}$ always exists, and
$I_{\partial}$ is uniquely determined from $I_{\cap}$.

As an example consider again the case of two variables (Figure~\ref{fig:2-vars}).  When $S$ is given, then the upper
three terms in the lattice correspond to the mutual informations $I(S:X_{1})$, $I(S: X_{2})$ and $I(S: X_{1}X_{2})$.
The lowest term, $I_{\cap}(S:X_{1};X_{2})$ is the shared information $SI(S:X_{1};X_{2})$.  The PI decomposition is
\begin{align*}
  I_{\cap}(S: \{X_{1}X_{2}\}) & = I_{\partial}(S: \{X_{1}\}\{X_{2}\}) + I_{\partial}(S:\{X_{1}\}) 
  \\ & \qquad + I_{\partial}(S:\{X_{2}\}) + I_{\partial}(S:\{X_{1}X_{2}\}), \\
  I_{\cap}(S:\{X_{1}\})      & = I_{\partial}(S: \{X_{1}\}\{X_{2}\}) + I_{\partial}(S:\{X_{1}\}), \\
  I_{\cap}(S:\{X_{2}\})      & = I_{\partial}(S: \{X_{1}\}\{X_{2}\}) + I_{\partial}(S:\{X_{2}\}), \\
  I_{\cap}(S:\{X_{1}\}\{X_{2}\}) & = I_{\partial}(S: \{X_{1}\}\{X_{2}\}).
\end{align*}
A comparison with~\eqref{equ:decomp} and~\eqref{equ:decomp2} shows that
\begin{align*}
  I_{\partial}(S:\{X_{1}X_{2}\})      & = CI(S: X_{1}; X_{2}), \\
  I_{\partial}(S:\{X_{1}\})          & = UI(S: X_{1} \setminus X_{2}), \\
  I_{\partial}(S:\{X_{2}\})          & = UI(S: X_{2} \setminus X_{1}), \\
  I_{\partial}(S: \{X_{1}\}\{X_{2}\}) & = SI(S: \{X_{1}\}\{X_{2}\}).
\end{align*}

As stated above, when $I_{\cap}$ is known, then $I_{\partial}$ can be computed uniquely using a Möbius inversion.  In
general, $I_{\partial}$ may have negative values.  In order to have a natural interpretation of the PI decomposition, we
need to require:
\begin{itemize}
\item[\textbf{(LP)}] $I_{\partial}\ge 0$.
  \hfill\emph{(local positivity)}
\end{itemize}
Local positivity can also be expressed as a condition on~$I_{\cap}$, see~\cite{WB}.

\section{Further Natural Properties of Shared Information}
\label{sec:further-properties}

The properties presented in the preceding section were identified by Williams
and Beer and are naturally related to the notion of the PI lattice.
Unfortunately, they are not enough to specify the function $I_{\cap}$ uniquely.
The properties are incomplete for mainly two reaons: First,
they do not tell us much about the left hand side apart from the normalization
condition~\textbf{(I)}.  Second, they do not tell us enough about what
happens when we add another argument on the right.

In this section we propose natural properties that describe the role of the
left-hand side.  Our first proposal is the following property:
\begin{itemize}
\item[\textbf{(S$_{1}$)}] $I_{\cap}(S:\mathbf{A}_{1};\dots;\mathbf{A}_{k})$ is symmetric in $S,\mathbf{A}_{1},\dots,\mathbf{A}_{k}$.
  \hfill\emph{(strong symmetry)}
\end{itemize}
In the following, we mostly consider the case that $S=\{X_{1},\dots,X_{n}\}$, and in this case \textbf{(M)} and
\textbf{(S$_{1}$)} together imply that $I_{\cap}(S:\mathbf{A}_{1};\dots;\mathbf{A}_{k}) =
I_{\cap}(\mathbf{A}_{1}:\mathbf{A}_{2};\dots;\mathbf{A}_{k})$, and hence we may omit the first argument~$S$.

Unfortunately, strong symmetry is not satisfied by many information theoretic
quantities that are used to quantify shared information or synergy,
but nevertheless we think that it is natural: If
$I_{\cap}$ has just two arguments, then strong symmetry does hold,
since the mutual information is symmetric.  In other words, the amount
of information that one random variable $X_{1}$ contains about another
variable $X_{2}$ is the same as the amount of information that $X_{2}$
carries about $X_{1}$.  It is natural to assume that an analogous
statement should hold if $I_{\cap}$ has more than two arguments.  Note
that the co-information $I_{Co}$ is symmetric in all its arguments.

Under the strong symmetry assumption, if we consider two variables $X_{1}$ and
$X_{2}$ and set $S=\{X_{1},X_{2}\}$, then all functions are fixed.  The
corresponding lattice is depicted in Figure~\ref{fig:2-vars}b).  We will see
later that, given the other properties, strong symmetry contradicts the local
positivity in the case of three random variables $X_1,X_2,X_3$.  The
implications of this will be discussed later.

A weaker property restricting the dependence on the first argument is the following:
\begin{itemize}
\item[\textbf{(LM)}] $I_{\cap}(S: \mathbf{A}_{1}; \dots; \mathbf{A}_{k})\le I_{\cap}(S S':\mathbf{A}_{1};\dots;\mathbf{A}_{k})$.
  \hfill\emph{(left monotonicity)}
\end{itemize}
This property captures the intuition that if $\mathbf{A}_{1},\dots,\mathbf{A}_{k}$ share some information about~$S$, then at least the
same amount of information is available to reduce the uncertainty about the joint outcome of $S$ and~$S'$.
Left monotonicity follows, of course, from monotonicity and strong symmetry.

Another property, which is independent from strong symmetry and which also implies \textbf{(LM)}, is the following:
\begin{itemize}
\item[\textbf{(LC)}] 
  $I_{\cap}(S S':\mathbf{A}_{1};\dots;\mathbf{A}_{k}) = I_{\cap}(S: \mathbf{A}_{1}; \dots; \mathbf{A}_{k})
  + I_{\cap}(S': \mathbf{A}_{1}; \dots; \mathbf{A}_{k} | S)$

  \hfill\mbox{\emph{(left chain rule)}}
\end{itemize}
where $I_{\cap}(S': \mathbf{A}_{1}; \dots; \mathbf{A}_{k} | S)$ is given by
$\sum_{s \in \set{S}} p(s) I_{\cap}(S': \mathbf{A}_{1}; \dots; \mathbf{A}_{k} |
s)$, i.e.  all distributions are conditioned on $s$ and then the average is
taken to obtain a conditional information.  This property is a natural
generalization of the chain rule of mutual information.  Moreover, a similar
property is used in Shannon's axiomatic characterization of entropy.

Unfortunately, the left chain rule is not fullfilled by any of the proposed
measures for shared information that we discuss later.
Nevertheless, we state it here, since we find it mathematically appealing.  The
same is true for left monotonicity: Most measures do not satisfy~\textbf{(LM)},
see Section~\ref{sec:Imin}.

The left chain rule together with local positivity also implies the following property
which has recently been proposed by~\cite{HarderSalgePolani12:Bivariate_Redundant_Information}:
\begin{itemize}
\item[\textbf{(Id$_{2}$)}] $I_{\cap}(\mathbf{A_{1}}\cup\mathbf{A_{2}}: \mathbf{A_{1}};\mathbf{A_{2}}) = I(\mathbf{A_{1}}:\mathbf{A_{2}})$.
  \hfill\mbox{\emph{(identity)}}
\end{itemize}
The identity property implies that $I_{\cap}(\{X_{1},X_{2}\}: X_{1};X_{2})$
vanishes if $X_{1}$ and $X_{2}$ are independent.  At first sight it seems
natural that independent random variables cannot share information.  However, in
Section~\ref{sec:game-theory} we will argue that they may indeed share information in this case.




\section{The Functions $I_{\min}$ and $I_{I}$}
\label{sec:Imin}

Williams and Beer define a function $I_{\min}(S,\mathbf{A}_{1},\dots,\mathbf{A}_{k})$ which satisfies all their
properties \textbf{(GP)}, \textbf{(S$_{0}$)}, \textbf{(I)} and \textbf{(M)} as follows:
\begin{align*}
  I_{\min}(S:\mathbf{A}_{1};\dots;\mathbf{A}_{k})
  &= \sum_{s}p(s) \min_{i} \sum_{a_i}p(a_i|s)\log\frac{p(s|a_i)}{p(s)} \\
  &= \sum_{s}p(s) \min_{i} \sum_{a_i}p(a_i|s)\log\frac{p(a_i|s)}{p(a_i)} \\
  &= \sum_{s} \min_{i} \sum_{a_i}p(a_i,s)\log\frac{p(a_i,s)}{p(a_i)p(s)}.
\end{align*}
The idea is the following: For each $i$ compare the prediction $p(s|a_{i})$ of
$S$ by $\mathbf{A}_{i}$ with the prior distribution $p(s)$ of~$S$.  Then combine
a minimization over $i$ with a suitable average using the joint distribution of
$\mathbf{A}_{i}$ and~$S$.

%

The order of the minimization and the averaging plays a crucial role.  If we interchange it, we obtain another function
\begin{equation*}
  I_{I}(S: A_{1};\dots;A_{k}) = 
  \min_{i} \sum_{s}p(s) \sum_{a_i}p(a_i|s)\log\frac{p(s|a_i)}{p(s)}
  = \min_{i}\{ I(S: A_{i}) \}\,.
\end{equation*}
This function $I_{I}$ satisfies the same properties, including
local positivity~\textbf{(LP)} (the proof of~\cite{WB} that proves \textbf{(LP)}
for $I_{\min}$ applies).  Of course, $I_{I}$ does not at all capture the intuition
behind the notion of shared information: $I_{I}$ just compares absolute values
of mutual informations, without caring whether different variables contain ``the same information.''
We will later argue that $I_{\min}$ suffers from a similar flaw (in particular,
$I_{\cap}=I_{I}$ in the examples considered below).  Note that any function
$I_{\cap}$ satisfying the properties \textbf{(GP)}, \textbf{(S$_{0}$)},
\textbf{(I)} and \textbf{(M)} satisfies $I_{\cap}\le I_{I}$.  In particular,
$I_{\min}\le I_{I}$.

The function $I_{I}$ satisfies left monotonicity.
However, $I_{\min}$ does not: For example, the following joint probability distribution
\begin{center}
  \begin{tabular}{cccc|l}
    $X_1$ & $X_2$ & $S$ & $S'$ &  \\ \hline
    0  &  0  &  0  &  0   &  1/6$^{\phantom{|}}$ \\
    0  &  1  &  0  &  0   &  1/6  \\
    0  &  1  &  0  &  1   &  1/6  \\
    1  &  1  &  0  &  1   &  1/6  \\
    1  &  0  &  1  &  1   &  2/6
  \end{tabular}
\end{center}
satisfies $I_{\min}(S: X_1; X_2) = \frac13 + \frac23(\frac34\log_2 3 - 1)
> I_{\min}(SS': X_1; X_2) = \frac13$.  This example can be understood as
follows: If $S=0$, then both $X_1$ and $X_2$ have some information
about~$S$ and thus contribute $\frac34\log_2 3 - 1$ bits to $I_{\min}$ in
this case.  However, if we additionally condition on~$S'$, then in any
case one of $X_1$ or $X_2$ carries no information: To be precise, if
$(S,S')=(0,0)$, then $X_2$ is uniformly distributed, and if
$(S,S')=(0,1)$, then $X_1$ is uniformly distributed.  Thus, in both
cases the minimization contributes zero bits to $I_{\min}$.  The remaining case
$(S,S')=(1,1)$ is equivalent to the case $S=1$, where both
$X_1$ and $X_2$ are fixed, and contributes one bit with weight~$\frac13$.

Omitting the calculations we mention that the redundancy measure proposed
by~\cite{HarderSalgePolani12:Bivariate_Redundant_Information} (and denoted by
$I_{HSP}$ in Section~\ref{sec:geometry}) also violates left monotonicity in the
same example.

\section{The Case of Three Variables}
\label{sec:three-variables}
\begin{figure}
  \centering
  \begin{minipage}{0.49\linewidth}
    a)\hspace{-1.5em}
    \xymatrix@C=1.4em{
      & 123 \ar[ld]\ar[d]\ar[rd] & & \\
      12 \ar[d]\ar[rd]      & 13 \ar[ld]\ar[rd]        & 23 \ar[ld]\ar[d] & \\
      12|13 \ar[d]\ar[rrrd] & 12|23 \ar[d]\ar[rrd]     & 13|23 \ar[d]\ar[rd] & \\
      1 \ar[d]              & 2 \ar[d]                 & 3 \ar[d] & 12|13|23 \ar[llld]\ar[lld]\ar[ld] \\
      1|23 \ar[d]\ar[rd]    & 2|13 \ar[ld]\ar[rd]      & 3|12 \ar[ld]\ar[d] & \\
      1|2 \ar[rd]           & 1|3 \ar[d]               & 2|3 \ar[ld] & \\
      & 1|2|3 & & }
  \end{minipage}
  \hfill
  \begin{minipage}{0.49\linewidth}
    b)\hspace{-1.5em}
    \xymatrix@C=1em{
      & H(123) \ar[ld]\ar[d]\ar[rd] & & \\
      H(12) \ar[d]\ar[rd]      & H(13) \ar[ld]\ar[rd]        & H(23) \ar[ld]\ar[d] & \\
      I(12:13) \ar[d]\ar[rrrd] & I(12:23) \ar[d]\ar[rrd]     & I(13:23) \ar[d]\ar[rd] & \\
      H(1) \ar[d]              & H(2) \ar[d]                 & H(3) \ar[d] & \textbf{?} \ar[llld]\ar[lld]\ar[ld] \\
      I(1:23) \ar[d]\ar[rd]    & I(2:13) \ar[ld]\ar[rd]      & I(3:12) \ar[ld]\ar[d] & \\
      I(1:2) \ar[rd]           & I(1:3) \ar[d]               & I(2:3) \ar[ld] & \\
      & \textbf{?} & & }
  \end{minipage}
  
  \caption{The PI lattice for $n=3$. For simplicity the sets are abbreviated by juxtaposing the indices of the corresponding variables.  For example, $12|13$ corresponds to $\{X_{1},X_{2}\}\{X_{1},X_{3}\}$.  a) The PI lattice.  b) The redundancies at the nodes, assuming strong symmetry and $S=\{X_{1},X_{2},X_{3}\}$.}
\label{fig:3-vars}
\end{figure}
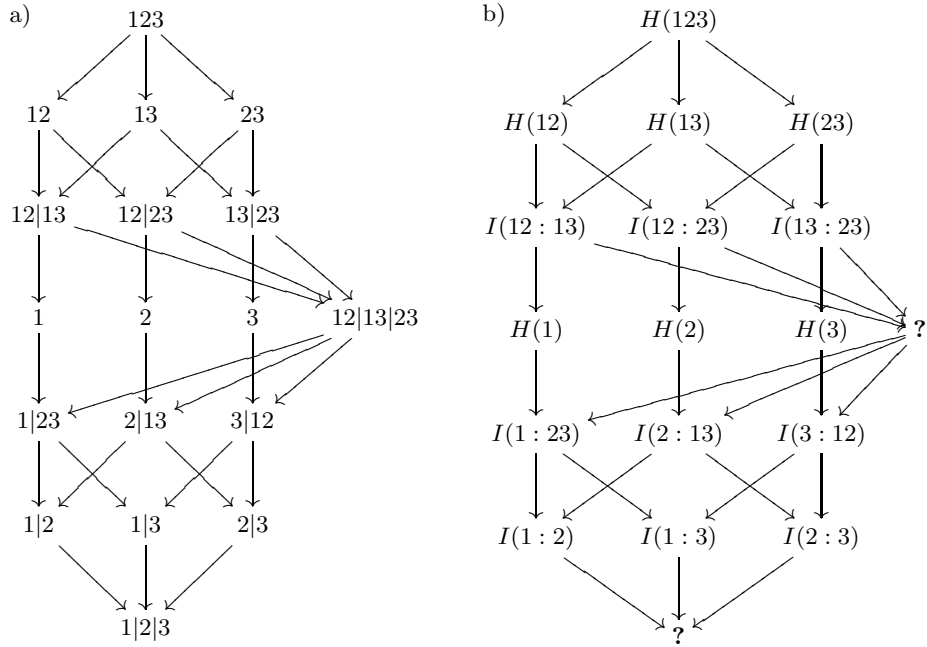

For three variables, the PI lattice is depicted in Figure~\ref{fig:3-vars}a).  Under the assumption of strong symmetry
all but two values in this lattice are fixed, see Figure~\ref{fig:3-vars}b).  The unknown values correspond to the
information shared by three random variables.

\begin{figure}
  \centering
  \begin{minipage}{0.37\linewidth}
    a)\hspace{-1.3em}
    \xymatrix@C=1em{   
      & 2(0) \ar[ld]\ar[d]\ar[rd] & & \\
      2(0) \ar[d]\ar[rd]   & 2(0) \ar[ld]\ar[rd] & 2(0) \ar[ld]\ar[d] & \\
      2(0) \ar[d]\ar[rrrd] & 2(0) \ar[d]\ar[rrd] & 2(0) \ar[d]\ar[rd] & \\
      1(0) \ar[d]          & 1(0) \ar[d]         & 1(0) \ar[d] & 2(1) \ar[llld]\ar[lld]\ar[ld] \\
      1(0) \ar[d]\ar[rd]   & 1(0) \ar[ld]\ar[rd] & 1(0) \ar[ld]\ar[d] & \\
      1(0) \ar[rd]         & 1(0) \ar[d]         & 1(0) \ar[ld] & \\
      & 1(1) & & }
  \end{minipage}
  \begin{minipage}{0.62\linewidth}
    b)\hspace{-1.5em}
    \xymatrix@C=0.8em{
      & H(123)=2 \ar[ld]\ar[d]\ar[rd] & & \\
      H(12)=2 \ar[d]\ar[rd]      & H(13)=2 \ar[ld]\ar[rd]        & H(23)=2 \ar[ld]\ar[d] & \\
      I(12:13)=2 \ar[d]\ar[rrrd] & I(12:23)=2 \ar[d]\ar[rrd]     & I(13:23)=2 \ar[d]\ar[rd] & \\
      H(1)=1 \ar[d]              & H(2)=1 \ar[d]                 & H(3)=1 \ar[d] & \textbf{?} \ar[llld]\ar[lld]\ar[ld] \\
      I(1:23)=1 \ar[d]\ar[rd]    & I(2:13)=1 \ar[ld]\ar[rd]      & I(3:12)=1 \ar[ld]\ar[d] & \\
      I(1:2)=0 \ar[rd]           & I(1:3)=0 \ar[d]               & I(2:3)=0 \ar[ld] & \\
      & 0 & & }
  \end{minipage}

  \caption{Redundancies in the XOR-example: a) $I_{\min}(123,\cdot)$ in the example.  The numbers in parentheses are
    $I_{\partial}(123,\cdot)$.  b) The shared information assuming strong symmetry.}
  \label{fig:xor}
\end{figure}
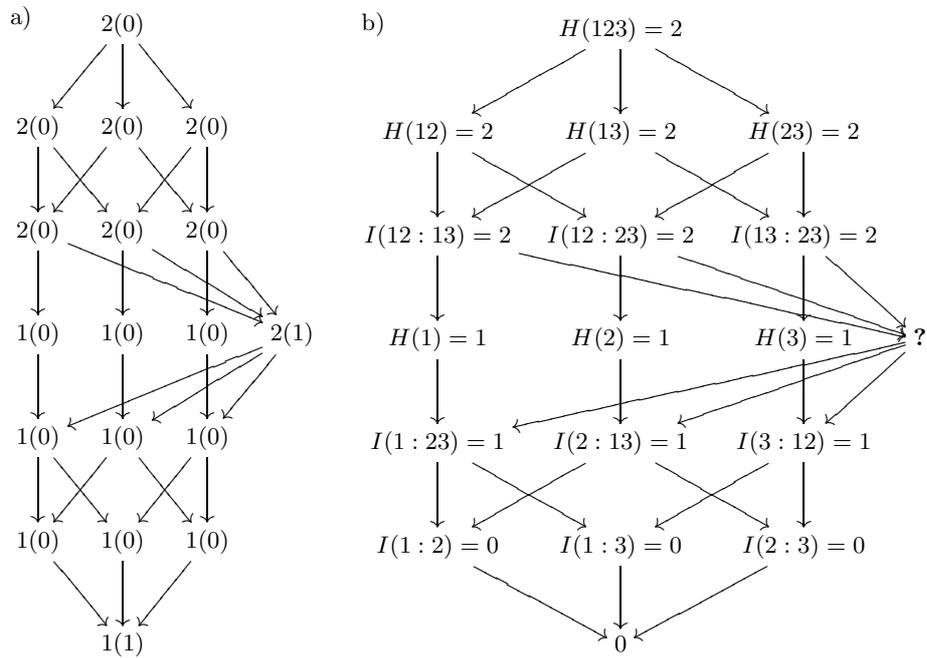
In the following, we discuss an example with three random variables $X_{1}$, $X_{2}$, $X_{3}$: Assume that $X_{1}$ and
$X_{2}$ are independent binary random variables, and let $X_{3}= X_{1}\oplus X_{2}$, where $\oplus$ denotes the sum
modulo 2 or the XOR-function.  Note that this example is symmetric in $X_{1}$, $X_{2}$ and $X_{3}$.
Figure~\ref{fig:xor}a) shows the values of $I_{\min}$ and $I_{\partial}$ in this example for $S=\{X_{1},X_{2},X_{3}\}$;
in other words, we decompose the information that the system has about itself.  What is striking is that the lowest
entry in this lattice does not vanish: According to $I_{\min}$, $X_{1}$, $X_{2}$ and $X_{3}$ share one bit of
information, although they are pair-wise independent. This fact that independent
variables may share information according to $I_{\min}$ has also been observed
and criticized in~\cite{HarderSalgePolani12:Bivariate_Redundant_Information}.  We will later
give an argument from game theory that explains how independent variables can
share information.  Nevertheless, in our opinion one bit of shared information
is too much in this situation: The absolute value of one bit of shared
information needs to be compared to the fact that each of $X_{1},X_{2},X_{3}$ does
not carry more than one bit of information.  Note that in the XOR-example
$I_{\min}=I_{I}$.

A close analysis of this also reveals that strong positivity is incompatible
with the PI lattice:
\begin{theorem}
  \label{thm:incompatibility}
  There is no measure of shared information that satisfies
  \textup{\textbf{(S$_{1}$)}}, \textup{\textbf{(M)}}, \textup{\textbf{(I)}}
  and~\textup{\textbf{(LP)}}.
\end{theorem}
\begin{proof}
  Assume that $I_{\cap}$ is a monotone function on the PI lattice that satisfies
  strong symmetry~\textbf{(S$_{1}$)}.  In the PI lattice for the XOR-example
  we can express all values on the lattice in terms of entropies and mutual
  informations, with one exception, see Figure~\ref{fig:xor}b).  Note that, by
  strong symmetry,
$I_{\cap}(X_{1} X_{2} X_{3}:\mathbf{A}_{1};\dots;\mathbf{A}_{k}) = I_{\cap}(\mathbf{A}_{1};\dots;\mathbf{A}_{k})$
whenever $\mathbf{A}_{1}\cup\dots\cup \mathbf{A}_{k}\subseteq\{X_{1},X_{2},X_{3}\}$.  Comparing with
Figure~\ref{fig:3-vars}b) we see that the information shared by $X_{1}$, $X_{2}$ and $X_{3}$ must vanish by
monotonicity, since the terms on the next layer also vanish, $I(X_{i},X_{j})=0$ for $i\neq j$.  Only the information
  shared by the pairs $\{X_{1},X_{2}\}$, $\{X_{1},X_{3}\}$ and $\{X_{2},X_{3}\}$
  is not determined.  However, we can bound these terms by the monotonicity.
  Similarly, we can compute bounds on $I_{\partial}$.  Namely,
\begin{multline*}
  I_{\partial}(\{X_{1},X_{2}\};\{X_{1},X_{3}\};\{X_{2},X_{3}\})
  = I_{\cap}(\{X_{1},X_{2}\};\{X_{1},X_{3}\};\{X_{2},X_{3}\}) \\
  - I_{\cap}(\{X_{1}\};\{X_{2},X_{3}\}) - I_{\cap}(\{X_{2}\};\{X_{1},X_{3}\}) - I_{\cap}(\{X_{3}\};\{X_{1},X_{2}\}) \pm 0 \\
  \le 2 - 3 = -1,
\end{multline*}
where $\pm 0$ represents a sum of terms belonging to the lowest two layers of
the PI diagram, and these terms all vanish.  This calculation shows that local
positivity is not possible.\qed
\end{proof}

To resolve this problem, one of the properties mentioned in
Theorem~\ref{thm:incompatibility} has to be dropped.
The easiest solution is to drop strong symmetry.  What are the alternatives? We
have to keep self-redundancy~\textbf{(I)} and local positivity~\textbf{(LP)},
since we want to find a decomposition of mutual-information into positive terms.
Therefore, if we want to keep strong symmetry,
we need to replace monotonicity~\textbf{(M)}.  It is probably a good idea to
keep the inequality condition in~\textbf{(M)}, but it is conceivable to replace
the equality condition.  However, one must keep in mind that the equality
condition is essential in justifying the use of the PI lattice: Without this
condition the values of the function $I_{\cap}$ on arbitrary collections of
subsets are not determined by its values on the antichains, and so the PI
lattice is not any more the natural domain of shared information.
Therefore, without the equality condition in~\textbf{(M)} we need to compute many
more terms to completely specify~$I_{\cap}$.  In turn, this means that there are
many more local terms~$I_{\partial}$.
With these additional terms it may be possible to obtain local positivity and
strong symmetry at the same time.
%

Heuristically, what happens in the XOR-example is the following: The term $I(X_{1}: X_{2}X_{3})$ on the third layer in
Figure~\ref{fig:3-vars} (counted from below) is equal to one bit, since we can compute $X_{1}$ from $X_{2}$ and $X_{3}$,
and hence $I(X_{1}: X_{2} X_{3})=H(X_{1})$.  Intuitively, the information shared between $X_{1}$ and $\{X_{2},X_{3}\}$
is precisely the information contained in $X_{1}$.  However, the three terms $I(X_{1}: X_{2} X_{3})$, $I(X_{2}: X_{1}
X_{3})$ and $I(X_{3}: X_{1} X_{2})$ on the third layer are not independent, since $X_{1}$, $X_{2}$ and $X_{3}$ are not
completely independent, but only pairwise independent.  Hence, if we compute the information shared by all three pairs,
we cannot just add up these three bits: We have to subtract (at least) one bit, which we overcounted.  Somehow this one
bit that we overcounted does not have a place in the PI lattice.

If we drop strong symmetry and keep the PI lattice, it is still the question how
to distribute the information over the PI lattice in the XOR-example.  In any
case, monotonicity implies that $I_{\cap}(S: X_{1} X_{2} ; X_{1} X_{3}; X_{2} X_{3})\le
I(S: X_{1} X_{2} X_{3})=2$.  On the other hand, the other three values on the third layer, the three mutual
informations $I(S: X_{i})$, are all equal to one bit.  These values restrict the possible values of $I_{\partial}$,
and it is not easy to motivate a non-negative assignment on intuitive grounds, even for this simple example.

\section{A Geometric Picture of Shared Information}
\label{sec:geometry}

One problem that makes it difficult to define shared information is that there is no known experimental way to extract
shared information.  In this section we want to assume that shared information can be extracted or modelled concretely.
We not only search for a number that measures the amount of shared information, but we want to represent the information
itself.

As a motivation consider the case of two random variables $X,Y$ from the perspective of coding theory.  Suppose that we
want to transmit information about $X$ and $Y$ over some channel.  Then the capacity that we need must exceed the amount of
information that we want to transmit.  To transmit a single variable~$X$, we need a capacity of~$H(X)$.  To be
precise, this statement only becomes true asymptotically: When we want to transmit a string of $n$ values of $n$
independent copies of~$X$, then, for large~$n$, if we have a channel with a capacity of $H(X)$ per time unit $\Delta T$,
then the time needed to transmit $X$ is roughly~$n\Delta T$.  In the same sense, to transmit $X$ and $Y$ together, we
need a channel of capacity~$H(\{X,Y\})$.  Suppose that $X$ was already transmitted, i.e.~both sender and receiver know the
value of~$X$.  As Shannon showed, in this case a channel of capacity $H(Y|X)=H(\{Y,X\})-H(X)$ is sufficient to transmit the
remaining information, such that the receiver knows both $X$ and~$Y$.  Hence, $H(Y|X)$ has the natural interpretation of
unique information of $Y$ with respect to $X$, and as Shannon's theorem shows, the unique information can be isolated
and transmitted separately.  The question is: Which other parts of information can be isolated?

As before, we consider information about a random variable~$S$.  We follow the paradigm that our information or belief
about $S$ can be encoded in a probability distribution~$p(S)$.  Suppose that $X$ is another random variable.  If $S$ is
not independent of~$X$, then a measurement of $X$ gives us further information about~$S$.  For example, if we know that
$X=x$, then our belief about $S$ can be encoded in the conditional probability distribution~$p(S|x)$.
Thus, the information that $X$ carries about $S$ can be encoded in a family $\{p(S|x)\}_{x \in \set{X}}$ of probability
distributions for~$S$. These distributions encode the {\em posterior} beliefs about $S$ conditioned on each outcome
of~$X$.

As motivated by Shannon, information can be quantified by logarithms of probabilities: The information that the state of
the variable $S$ is equal to the specific value $s$ is worth $-\log_{2}(p(S=s))$.  Our uncertainty about $S$, when our
knowledge is encoded in the distribution~$p(S)$, is then equal to the expected information gain when we learn the value
of $S$:
\begin{equation*}
  \sum_{s} p(s)(-\log(p(s))) =: H(S).
\end{equation*}
Similarly, the information that we gain when we learn that $X=x$ is equal to the conditional entropy $H(S|X=x) =
-\sum_{s} p(s|x)\log(p(s|x))$.  The (expected) information that $X$ brings us about $S$ is obtained by averaging
$H(S|X=x)$ and comparing the value with $H(S)$; this agrees with the mutual information:
\begin{multline*}
  \sum_{x}p(x) \sum_{s} p(s|x) \log(p(s|x)) - \sum_{s} p(s)\log(p(s))
  \\
  = \sum_{x}p(x)\sum_{s}p(s|x) \log\left(\frac{p(s|x)}{p(s)}\right) = I(X:S).
\end{multline*}



The situation can be pictured geometrically.  Let $\Pcal_{S}$ be the set of all probability distributions for~$S$.
Geometrically,
\begin{equation*}
\Pcal_{S}=\left\{p: \set{S} \to \Rb : p(s) \ge 0, \sum_{s \in \set{S}}p(s)=1 \right\}
\end{equation*}
is a simplex.  The family $\{p(S|x)\}_{x}$
is a point configuration in $\Pcal_{S}$, indexed by the outcomes $x$ of the random variable $X$.
The information gain is then the mean reduction of uncertainty (in the sense of Shannon information) when replacing the
prior~$p(S)$ with the family $\{p(S|x)\}_{x}$.

According to our geometric interpretation of information, the shared information that $X_{1},\dots,X_{k}$ carry about
$S$ should also be representable as a weighted family of probability distributions for~$S$.  The question is how to
construct this weighted family from the posteriors $\{p(S|x_{i})\}_{x_i}$ and the joint distribution of
$X_{1},\dots,X_{k}$ and~$S$.  Suppose that we have found such a family representing the shared information, and denote
it by $\{p_{x_{1}\sh x_{2}\sh \dots\sh x_{k}}(S)\}_{x_{1},\dots,x_{k}}$.
Then we want to quantify the shared information.  There are two natural possibilities:
\begin{align*}
  SI_{lr}(S:X_{1};\dots;X_{k}) &:= \sum_{x_{1},\dots,x_{k}}\sum_{s} p(s,x_{1},\dots,x_{k})
                 \log\left(\frac{p_{x_{1}\sh x_{2}\sh \dots\sh x_{k}}(s)}{p(s)}\right)
\\
  SI_{KL}(S:X_{1};\dots;X_{k}) &:= \sum_{x_{1},\dots,x_{k}} p(x_{1},\dots,x_{k})
  D(p_{x_{1}\sh x_{2}\sh \dots\sh x_{k}}\|p).  
\end{align*}
The function $SI_{KL}$ has the advantage that it always satisfies global
positivity, regardless of how we construct $p_{x_{1}\sh x_{2}\sh \dots\sh
  x_{k}}$.  By contrast, the function $SI_{lr}$ directly measures the change of
surprise when we replace the prior distribution $p(s)$ with the distribution
$p_{x_{1}\sh x_{2}\sh \dots\sh x_{k}}$.  Depending on how we construct
$p_{x_{1}\sh x_{2}\sh \dots\sh x_{k}}$ the value of $SI_{lr}$ may become
negative.

 We would like to have the following properties:
\begin{enumerate}
\item The construction should be symmetric in $x_{1},\dots,x_{k}$.
\item If $k=1$, then we obtain the posterior: $p_{x_{1}}(S) = p(S|x_{1})$.
\item More variables share less information:

  $D(p_{x_{1}\sh\dots\sh x_{k}}(S)\|p(S)) \le D(p_{x_{1}\sh \dots\sh x_{k-1}}(S)\|p(S))$.
\end{enumerate}
These properties are related to the properties {\bf (S), (I)} and {\bf (M)} as stated above,
but re-formulated to hold point wise for each joint outcome $x_1, \ldots, x_k$.

A natural candidate satisfying the above properties is given by
\begin{equation*}
  p_{x_{1}\sh \dots\sh x_{k}}(S)
  = \argmin \left\{ D(\sum_{i=1}^{k}\lambda_{i}p(S|x_{i})\| p(S)) \;:\; \lambda_{i}>0, \sum_{i}\lambda_{i}=1\right\}.
\end{equation*}
Since the KL divergence is convex, the function $p\mapsto D(p\|p(S))$ has a
unique minimum on any closed convex set.  This shows that the above definition
is well-defined.  Moreover, the definition ensures that $p_{x_{1}\sh \dots\sh
  x_{k}}(S)$ belongs to the convex hull of the posteriors $p(S|x_{i})$ for
$i=1,\dots,k$.  This models the fact that $p_{x_{1}\sh \dots\sh x_{k}}(S)$ only
involves information that is present in these posteriors.  In fact, $p_{x_{1}\sh
  \dots\sh x_{k}}$ is the least informative distribution from this convex set.



The construction of $p_{x_{1}\sh \dots\sh x_{k}}(S)$ implies the following property, which gives an idea in which sense
$p_{x_{1}\sh \dots\sh x_{k}}(S)$ summarizes information shared among all the posteriors $p(S|x_{i})$:
\begin{lemma}
  If all $p(S|x_{i})$ satisfy some linear inequality, then $p_{x_{1}\sh \dots\sh x_{k}}(s_{2})$ satisfies the same
  inequality.  In particular:
  \begin{enumerate}
  \item If $p(s_{1}|x_{i}) \le p(s_{2}|x_{i})$ for all $i$, then $p_{x_{1}\sh \dots\sh x_{k}}(s_{1}) \le p_{x_{1}\sh
      \dots\sh x_{k}}(s_{2})$.
  \item If $p(s|x_{i}) = 0$ for all $i$, then $p_{x_{1}\sh \dots\sh x_{k}}(s) = 0$.
  \end{enumerate}
\end{lemma}

Unfortunately, $SI_{lr}$ violates monotonicity, and with $SI_{KL}$ the synergy
can become negative.  Both facts can be illustrated with the same example:

From~\eqref{equ:coinformation-diff} we find that
\[ CI(S: X_1; X_2) = I(S: X_1|X_2) - I(S:X_1) + SI(S: X_1; X_2) \]
and thus the non-negativity of $CI$ requires that
\begin{equation}
  \label{eq:SI-ge-Ico}
  SI(S: X_1; X_2) \geq I(S: X_1) - I(S: X_1 | X_2) = I_{Co}(S: X_{1}:X_{2})\,.
\end{equation}
Now, if $S$ is a function of~$X_{2}$, then $I(S:X_{1}|X_{2})$ vanishes, and
therefore~\eqref{eq:SI-ge-Ico} implies $SI(S: X_1; X_2) \geq I(S: X_1)$.
Together with~\textbf{(M)} we obtain
\begin{equation}
  \label{eq:SI=MI}
  SI(S: X_1; X_2) = I(S: X_1),\quad\text{ if $S$ is a function of $X_{2}$.}
\end{equation}

Consider the following distribution
\begin{center}
\begin{tabular}{ccc|c}
  $s$ & $x_1$ & $x_2$ & $p(s,x_1,x_2)$ \\ \hline
  0 & 0 & 0 & 2/6 \\
  0 & 1 & 0 & 1/6 \\
  1 & 0 & 1 & 1/6 \\
  1 & 1 & 1 & 2/6
\end{tabular}
\end{center}
The relative location of $p(S)$ and the posteriors of $S$ given one or two of
$X_{1}$ and $X_{2}$ is visualized in Figure~\ref{fig:SIlr-SIKL}.  Under this
distribution $S$ and $X_{1}$ are positively correlated, while $S = X_2$, and
thus $I(S: X_1 | X_2) = 0$.  Consider the case $X_{1}=x_{1}\neq x_{2}=X_{2}$ in
which $X_{1}$ and $X_{2}$ have conflicting posterior about~$S$, i.e.~$p(S|x_2)$
assigns probability one to $S = x_2$, whereas $p(S|x_1)$ assigns a higher
probability to $S = x_1 \neq x_2$.  Thus, $p_{x_{1}\sh x_{2}}(S)$ is equal to
the prior $p(S)$ in this case.  On the other hand, if $X_{1}=X_{2}=x$, then both
posteriors favor $S=x$.  The convex hull of $p(S|x_{1})$ and $p(S|x_{2})$ is an
interval, and the posterior $p(S|x_{1})$ is the closest point to the
prior~$p(S)$.  Therefore, $p_{x_{1}\sh x_{2}}(S) = p(S|x_{1})$.  In total,
\begin{align*}
  I(S: X_{1}) - &SI_{KL}(S: X_{1}; X_{2}) \\
  & = \sum_{x_{1},x_{2}} p(s,x_{1},x_{2})\left(D_{KL}(p(S|x_{1})\|p(S)) - D_{KL}(p_{x_{1}\sh x_{2}}(S)\|p(S))\right)
  \\
  & = \sum_{x_{1}\neq x_{2}} p(s,x_{1},x_{2}) D_{KL}(p(S|x_{1})\|p(S)) > 0,
\end{align*}
and therefore~\eqref{eq:SI=MI} is violated.  One can check that in this case
$SI_{lr}(S: X_{1};X_{2})$ also violates~\eqref{eq:SI=MI}, but in the other
direction.  Therefore, $SI_{lr}$ violates monotonicity.
\begin{figure}
  \centering
  \begin{tikzpicture}[scale=9]
    \path (0,0) coordinate (S0);
    \path (1,0) coordinate (S1);
    \path (1/2,0) coordinate (Prior);
    \path (1/3,0) coordinate (X0);
    \path (2/3,0) coordinate (X1);

    \foreach \i in {0,1} { \fill (S\i) circle (0.3pt); }
    \foreach \i in {0,1} { \fill (X\i) circle (0.3pt); }
    \fill (Prior) circle (0.3pt);
    \draw (S0) -- (S1);
    \node[anchor=south] at (S0) {$p(S|X_2=0)$};
    \node[anchor=north] at (S0) {$\delta_{S=0}$};
    \node[anchor=south] at (S1) {$p(S|X_2=1)$};
    \node[anchor=north] at (S1) {$\delta_{S=1}$};
    \node[anchor=south] at (X0) {$p(S|X_1=0)$};
    \node[anchor=south] at (X1) {$p(S|X_1=1)$};
    \node[anchor=south] at (Prior) {$p(S)$};

    \begin{scope}[opacity=0.5]
      \draw[cap=round,green,line width=3pt] (S0)--(X0); \draw[cap=round,red,line
      width=3pt] (X0)--(S1);
    \end{scope}
  \end{tikzpicture}

  \caption{The construction of $p_{x_1\sh x_2}$ for the example to $SI_{KL}$
    and~$SI_{lr}$.  The set of probability distributions of the binary variable
    $S$ is the interval between the two point measures $\delta_{S=0}$ and
    $\delta_{S=1}$.  The convex hull of $p(S|X_{1}=0)$ and $p(S|X_{2}=0)$ is marked in
    green.  The closest point to the prior is $p(S|X_{1}=0)$.  The convex hull
    of $p(S|X_{1}=0)$ and $p(S|X_{2}=1)$ is marked in red; it contains the prior.}
  \label{fig:SIlr-SIKL}
\end{figure}
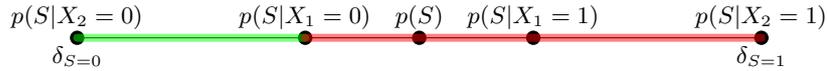

The geometric strategy pursued in this section 
can be compare with the strategy by Williams and Beer in~\cite{WB} that leads to
the definition of~$I_{\min}$.  The formula
\begin{multline*}
  I_{\min}(S:\mathbf{A}_{1};\dots;\mathbf{A}_{k})
  = \sum_{s}p(s) \min_{i} \sum_{a_i}p(a_i|s)\log\frac{p(a_i|s)}{p(a_i)}
  \\
  = \sum_{s}p(s) \min_{i} D(p(\mathbf{A_{i}}|S)\|p(\mathbf{A_{i}}))
\end{multline*}
defining $I_{\min}(S;\mathbf{A_{1}};\dots;\mathbf{A_{n}})$ is similar to the
defining equation of $SI_{KL}$, but involves the conditional distributions
$p(a_{i}|s)$ of the input given the output~$S$.  In our opinion it is much more
natural to work with distributions over the output variable~$S$, since, after
all, we are interested in information about~$S$.  Of course, the defining
equation of $I_{\min}$ can be rewritten in the form
\begin{equation*}
  I_{\min}(S:\mathbf{A}_{1};\dots;\mathbf{A}_{k})
  = \sum_{s}p(s) \min_{i} \sum_{a_i}p(a_i|s)\log\frac{p(s|a_i)}{p(s)},
\end{equation*}
which resembles the definition of~$SI_{lr}$, but involves minimizing over the inputs.

The proposed definition of the posteriors $p_{x_{1}\sh \dots\sh x_{k}}(s)$
involves similar ideas as the definition of shared information 
in~\cite{HarderSalgePolani12:Bivariate_Redundant_Information}.  We only sketch
these connections and refer to the
manuscript~\cite{HarderSalgePolani12:Bivariate_Redundant_Information} for the
precise definitions.  To distinguish their function from other functions we call
it~$I_{HSP}$.  The definition of $I_{HSP}(S : X_{1}; X_{2})$ involves approximating
the posteriors $p(s|x_{1})$ by the convex hull family of posteriors $p(s|x_{2})$
for all possible values $x_{2}$ of~$X_{2}$.  However, as defined
in~\cite{HarderSalgePolani12:Bivariate_Redundant_Information} this
approximation, denoted by $p_{(x_{1}\searrow X_{2})}(s)$, is not unique.  Then
\begin{multline*}
  I_{HSP}(S:X_{1};X_{2})
  = \min\bigg\{
    \sum_{s,x_{1}} p(s,x_{1}) \log\frac{p_{(x_{1}\searrow X_{2})}(s)}{p(s)},
    \\[-2mm]
    \sum_{s,x_{2}} p(s,x_{2}) \log\frac{p_{(x_{2}\searrow X_{1})}(s)}{p(s)}
  \bigg\}.
\end{multline*}

Note that in both definitions of $p_{(x_{1}\searrow X_{2})}(s)$ and $p_{x_{1}\sh
  \dots\sh x_{k}}(s)$ the notion of the convex hull is used as a means to
describe the set of distributions that involve information contained in a set of
posterior distributions.  The difference between both approaches is
that~\cite{HarderSalgePolani12:Bivariate_Redundant_Information} do not try to
extract and represent the joint information pointwise, but they try to model the information
contained in $X_{1}$ using the posterior distributions of~$X_{2}$.  This breaks
the symmetry, and therefore, in the end, one has to take a minimum. Furthermore,
this definition is only meaningful in the case of two random variables and violates
the left monotonicity (see Section~\ref{sec:Imin}).

\section{Game Theoretic Intuitions}
\label{sec:game-theory}

Without an operational definition it is hard to decide which of the above
properties and geometric structures are best suited to capture the concept of
shared information.  In order to get a better idea of what is actually meant
when talking about shared information, we highlight some aspects from the
perspective of game theory.

Scientists in both game theory~\cite{Aumann76} and computer
science~\cite{Halpern95} have studied how knowledge is distributed among a group
of agents.  Since knowledge can be regarded as certain information, results from
these disciplines can provide additional insights into shared information.  The
basic formalism of epistemic agents considers a set $\set{S}$ of possible states
of the world or situations. The knowledge of an agent $i$ is represented as a
partition $X_i$ on $\set{S}$.  Such a partition can be considered as a function $X_i:
\set{S} \to \set{X}_i$ mapping states of the world to possible observations
$\set{X}_i$ that are available to the agent\footnote{Note the similarity to the definition of a random variable as a measurable map
from a probability space to outcomes.  In fact, if we choose an arbitrary
probability distribution on~$\set{S}$, then the partition $X_{i}$, considered as
a function $\set{S}\to\set{X}_{i}$, becomes a random variable.}. Thus, each agent $i$ might not be able to observe the actual
state $s$ of the world, but given an observation $x_i$ he considers all
situations in $X_i^{-1}(x_i) = \{ s \in \set{S} \,|\, X_i(s) = x_i \}$ to be
possible.

Suppose that agent $i$ observes $x_i \in \set{X}_i$. Then $i$ is said to know an event, corresponding to a subset $E
\subset \set{S}$, if the event occurs in all situations that the agent holds possible
given $x_i$, i.e.
\begin{equation} 
  X_i^{-1}(x_i) \subseteq E. \nonumber
\end{equation}
This gives rise to the knowledge operators $K_i: 2^{\set{S}} \rightarrow
2^{\set{S}}$ taking an event $E$ to all situations where agent $i$ knows
this event:
\begin{equation}
  K_i(E) = \{ s \in S \;|\; \mbox{agent $i$ knows $E$ given the observation $X_i(s)$} \}.
  \label{equ:Ki}
\end{equation}
$K_{i}(E)$ can itself be considered as an event.
Using this operator~$K_{i}$, we can compute the situations where an event $E$ is
shared knowledge between agents $1, \ldots, n$, i.e.~where every agent
knows~$E$:
\[ SK(E) = \bigcap_{i=1}^n K_i(E) \]

Note that this does not imply that every agent knows that every agents
know $E$.  The much stronger requirement that everyone knows $E$, and
everyone knows that everyone knows this, and so on, is formalized by
iterating the above construction and referred to as \emph{common
  knowledge}:
\begin{equation*}
  CK(E) = \bigcap_{k=1}^{\infty}SK^{k}(E),\text{ where } SK^{k}(E) = (SK(\cdots SK(E)\cdots)) \text{ ($k$ iterations)}.
\end{equation*}


As an example consider the case of three binary random variables $X_1,X_2$ and~$S$,
where $X_1$ and $X_2$ are independent and $S$ consists of a copy of both of them.
Then, the set of possible situations, i.e.~the support of the joint distribution
$p(x_1,x_2,s)$, consists of four possible states:
\begin{center}
\begin{tabular}{ccc}
$X_1$ & $X_2$ & $S$ \\ \hline
0 & 0 & 00 \\
0 & 1 & 01 \\
1 & 0 & 10 \\
1 & 1 & 11
\end{tabular}
\end{center}
The information partitions correspond to the projections on the
respective components of the joint state, e.g.
\begin{align*}
  X_1^{-1}(0) &= \{ (0,0,00), (0,1,01) \},\\
  X_2^{-1}(1) &= \{ (0,1,01), (1,1,11) \}.
\end{align*}
For the event $E = \{ (0,0,00), (0,1,01), (1,0,10) \}$ we find that
\begin{align*}
  K_{1}(E) & = \{ (0,0,00), (0,1,01) \} \\
  K_{2}(E) & = \{ (0,0,00), (1,0,10) \},
\end{align*}
and therefore $SK_{1,2}(E) = \{ (0,0,00) \}$ since both agents $1$ and $2$
can exclude the state $(1,1,11)$ in this case. Thus, we conclude that
there exists non-trivial shared information between $X_1$ and $X_2$,
namely that $S \neq 11$, even though $X_1$ and $X_2$ are independent of each
other and neither of them knows the state of the other.
On the other hand, there is no common knowledge between $X_1$ and $X_2$, since
$SK_{1,2}(SK_{1,2}(E)) = \emptyset$.

Note that $I_{\min}(S : X_1; X_2) = I_{I}(S : X_1; X_2) = 1$ bit in this
example, if we assume that $X_1$ and $X_2$ are independent and uniformly distributed.
If we say that $I_{\min}$ measures the shared information, then
this implies that $X_1$ and $X_2$ have no unique information.
This is surprising, given that $X_1$
and $X_2$ are independent.
Regarding the game theoretic analysis we see that the shared knowledge only
rules out one state. Thus, a reasonable definition of shared information might
give a positive value to $I_{\cap}(S: X_1; X_2)$ even if $I(X_1: X_2) = 0$, but should
certainly stay below 1~bit.
Maybe a value of $\log(4/3)$ would be a good idea, since the number of
possibilities is reduced from four to three.
Note that \textbf{(Id$_{2}$)}, as proposed
in~\cite{HarderSalgePolani12:Bivariate_Redundant_Information}, would require
that $I_{\cap}(S: X_1;X_2) = 0$ whenever $I(X_1:X_2)$ vanishes.


At present, it is not clear how the difference between shared and common
information could be formulated in information theoretic terms.  One may also
ask, whether a desired decomposition of information, should take into account
shared information or rather refer to common information.  It would probably be
easier to use shared information in a decomposition, because otherwise one needs
to decompose the information into terms describing the information that $X_1$
knows that $X_2$ knows, but $X_2$ does not know whether it is known by $X_1$, and so
on.  On the other hand, common knowledge is represented as a partition (see \cite{Aumann76}),
and hence corresponds to a random variable after introducing a probability
measure on~$\set{S}$.  In contrast, shared knowledge cannot be represented as a
partition.  Maybe this explains why it is difficult, and may even be impossible,
to represent shared information as a random variable.

Note that the condition \textbf{(Id$_{2}$)} 
takes into account the mutual information between elements $\mathbf{A_{i}}$ of
the right hand side.  Their relationship is not considered in the definition of
shared knowledge, but only appears in the higher-order terms which are iterated
in the case of common knowledge.  Therefore, the property \textbf{(Id$_{2}$)} is
more natural for common information than for shared information.  The same holds
true for \textbf{(LC)}, since \textbf{(LC)} implies \textbf{(Id$_{2}$)}.

\section{Conclusions}

We have discussed natural and intuitive properties that a measure of shared information should have.  We have shown
that some of these properties contradict each other.  This shows that intuition and
heuristic arguments have to be used with great care when arguing about information.

In particular, we discussed the partial information decomposition and lattice introduced by Williams and Beer.  We have
shown that a positive decomposition according to the PI lattice contradicts another desirable property, called strong
symmetry.  We are unsure whether this is an argument against strong symmetry, or whether the PI lattice has
to be refined, since 
it is difficult to assign plausible values to the PI decomposition for the XOR-example.

Williams and Beer also proposed a concrete measure $I_{\min}$ of shared
information.  We show that in some examples this measure yields unreasonably
large values.  The problem is that $I_{\min}$ does not distinguish whether
different random variables carry \emph{the same} information or just \emph{the
  same amount} of information.  This phenomenon has also been observed by
others.  However, most people focussed on the property that independent
variables may share information about themselves.  We argue, using ideas from
game theory, that this fact in itself does not speak against~$I_{\min}$; but we
agree that the absolute value that $I_{\min}$ assigns to the shared information
is too large.  In our opinion, what is more striking, is that $I_{\min}$ is not
monotone in its left argument: Random variables share less information about
more.


We expect that further progress requires a more precise, operational
idea of what shared information should be.
We believe that our results provide additional
insights, even thought we have mainly revealed pitfalls regarding
the notion of shared information. Thus, despite some recent progress,
the quest for a general decomposition of multi-variate information is
still open.

\subsubsection*{Acknowledgments}

We thank Nihat Ay for stimulating discussions.
This work was supported by the VW Foundation (J.R.)  and has received funding from the European
Community's Seventh Framework Programme (FP7/2007-2013) under grant agreement no. 258749 (to EO). 

\bibliographystyle{splncs}
\bibliography{ECCS,NIPS}

\end{document}